\documentclass[usletter, 12pt]{article}

\usepackage{amsmath,amsthm}
\usepackage{natbib}
\usepackage{geometry}
\usepackage{amssymb,authblk,amsfonts}%,amsthm}
\usepackage{comment,enumerate,graphicx,hhline}

\usepackage[usenames,dvipsnames]{color, xcolor}

\usepackage{subcaption}
\usepackage[hidelinks]{hyperref}

\usepackage{amsmath,amsfonts,amstext,rotating,amssymb,graphics,xspace}

\usepackage{epsfig}
\usepackage{hyperref}
\hypersetup{
colorlinks   = true, %Colours links instead of ugly boxes
urlcolor     = blue, %Colour for external hyperlinks
linkcolor    = blue, %Colour of internal links
citecolor   = blue    %Colour of citations
}

\usepackage{amsmath}
\usepackage{natbib}
\usepackage{hyperref}
\usepackage{geometry}
\usepackage{type1cm}
\usepackage{mathrsfs}
\usepackage{amssymb}
\usepackage{authblk}
\usepackage{amsfonts}
\usepackage[T2A]{fontenc}
\usepackage{hhline}
\usepackage{amsthm}
\usepackage{multirow}
\usepackage{multicol}
\usepackage{color}
\usepackage{sgame}
\usepackage{mathtools}
\usepackage{caption}
\usepackage{graphicx}
\usepackage{wrapfig}
\usepackage{changepage}
\usepackage{lipsum}
\usepackage[english]{babel}
\usepackage[utf8]{inputenc}
\usepackage{comment}
\usepackage{tikz}
\usepackage{pgf}
\usetikzlibrary{positioning}
\usepackage{tabularx}
\usepackage{pbox}
\usepackage{subcaption}
\usepackage{enumitem}
\usepackage{rotating}
\usepackage{tcolorbox}
\usepackage{setspace}
\usepackage{array, makecell}

\usetikzlibrary{arrows,backgrounds}

\tikzstyle{object}=[circle,draw=red]
\tikzstyle{agent}=[circle,draw=blue]
\tikzstyle{quantity}=[fill=white]

\def\w{\omega}

\def\bx{\mathbf{x}}
\def\G{\Gamma}

\hypersetup{colorlinks = true}
\hypersetup{allcolors=blue}

\newtheorem{definition}{Definition}

\newtheorem{proposition}{Proposition}
\newtheorem{lemma}{Lemma}
\newtheorem{example}{Example}
\newtheorem{corollary}{Corollary}

\newcommand{\norm}[1]{\| #1 \|}

\usepackage{setspace}
\begin{document}
	
	\title{Efficient and fair trading mechanisms on the full preference domain}

\author{Jingsheng Yu\thanks{School of Economics and Management, Wuhan University. Email: yujingsheng1987@outlook.com} \quad \quad Jun Zhang\thanks{Institute for Social and Economic Research, Nanjing Audit University. Email: zhangjun404@gmail.com}
	}

\date{\today}

\maketitle
	
	\begin{abstract}\label{abstract}
    \cite{YuzhangFTTCnew} introduce a new method for defining trading mechanisms in market design and apply it to develop new mechanisms that achieves efficiency and fairness in various models. However, their assumption of strict preferences limits its applicability. This paper extends their approach to the full preference domain, allowing for indifferences while preserving key efficiency and fairness properties. As a corollary, we generalize the probabilistic serial mechanism in the house allocation model to the full preference domain without relying on operations research tools.		
	\end{abstract}
	
	
	%Absorbing set in Leontief model.
	
\thispagestyle{empty}
\setcounter{page}{0}	
	\newpage

\section{Introduction}\label{section:intro}

Trading mechanisms are widely used in market design to find efficient allocations. Whenever an allocation is inefficient, there exist opportunities for agents to exchange their assignments to improve efficiency. The traditional method for defining trading mechanisms originates from the definition of the top trading cycle (TTC) mechanism \citep{ShapleyScarf1974}. In this method, agents and objects, represented as nodes in a directed graph, point to each other according to preferences and ownership. Nodes forming a cycle then exchange objects among themselves. This method is desirable in specifying the details of the trading process, including who trades with whom and how trades occur. However, in complex environments where the resulting graphs are intricate, this method can be challenging to use, especially when addressing fairness considerations that require the integration of randomization into a trading process without compromising efficiency. To address this challenge, \cite{YuzhangFTTCnew} introduce a method that utilizes parameterized linear equations to define a trading process. Unlike the traditional approach, this method characterizes only the outcome of a trading process, that is, who obtains what and supplies what. This abstraction avoids complications that can arise in complex environments. Moreover, by adjusting the equation parameters, the method allows for flexible control of the trading process and the achievement of fairness. The authors apply this method to develop new mechanisms that achieve efficiency and fairness in various models, including fractional endowment exchange, priority-based allocation, and house allocation with existing tenants.

Many papers in the literature, including \cite{YuzhangFTTCnew}, assume that agents have strict preferences. Although this assumption accommodates many applications, there are environments where agents naturally have weak preferences.\footnote{See \cite{bogomolnaia2005strategy} and \cite{erdil2017two} for discussions on weak preferences.} To preserve efficiency, preference ties cannot be arbitrarily broken before executing a trading mechanism. This paper introduces a method to generalize the trading mechanisms developed by \cite{YuzhangFTTCnew} to the full preference domain in which agents may have indifferent preferences, while preserving desirable efficiency and fairness properties of those mechanisms.

A simple example illustrates why preference ties cannot be arbitrarily broken. Consider a fractional endowment exchange problem with two agents, $ i $ and $ j $, who each own equal probability shares of two objects, $ a $ and $ b $; that is, they each own $ (1/2a, 1/2b) $. Agent $ i $ is indifferent between $ a $ and $ b $, while $ j $ strictly prefers $ a $ to $ b $. If the ties are broken such that $ i $ strictly prefers $ a $ to $ b $, then in the unique individually rational allocation, both agents will retain their endowments. However, the only efficient and individually rational allocation that respects their true preferences is the one in which $ i $ receives $ b $ and $ j $ receives $ a $; that is, $ i $ supplies $ 1/2a $ to $ j $ in exchange for $ 1/2b $ from $ j $.

Our method for solving indifferent preferences is to utilize the trading process. At any step of a trading mechanism, if an agent receives a positive amount of an object that has been exhausted in the trading process but finds indifferent objects among those still available, we let the agent label her received amount of that object as an endowment that is available for trading. In subsequent steps, if the agent loses some amount of that object, she will be compensated with an equal amount of indifferent objects through the trading process. This ensures that the agent is not harmed but may benefit others who strictly prefer the exhausted object over the other available objects. In the above example, if the ties are broken such that $ i $ strictly prefers $ a $ to $ b $, then at step one, both $ i$ and $j $ receive their own endowments $ 1/2a $ through self-trade. After that, $ a $ is exhausted. At step two, since $i$ is indifferent between $a$ and $b$ and $ b $ is still available,  $ i $ labels her consumed $ 1/2a $ as available for trading. Then, $ j $ demands $ i $'s $ 1/2a $, and $ i $ demands $ b $. In the trading outcome, $ j $ receives $ i $'s $ 1/2a $, and $ i $ receives $ b $, consisting of her own $ 1/2b $ and $ j $'s $ 1/2b $.% This produces the unique desirable allocation.

The above example illustrates a simple situation. In general, after an agent labels her received amount of an object as available, it may induce another agent to label her received amount of a different object as available when the latter agent is indifferent between the two objects. This can trigger a chain of sequential labeling behavior. To preserve the efficiency of a trading mechanism, it is crucial to identify all such chains during the algorithm. Therefore, we add a labeling stage at each step of a trading mechanism to accomplish this.

We apply our method to extend the balanced trading mechanisms (BTM) developed by \cite{YuzhangFTTCnew} for the fractional endowment exchange model. The model generalizes the housing market model by allowing agents to own fractional endowments. We prove that these extensions preserve efficiency and fairness properties that hold on the strict preference domain. Our extensions of BTM can be similarly adapted to obtain extensions of the trading mechanisms developed by \cite{YuzhangFTTCnew} for the priority-based allocation model and for the house allocation with existing tenants model. For brevity, we omit the details for these two models. Instead, we focus on the house allocation model in the second part.

\cite{YuzhangFTTCnew} show that in the house allocation model with strict preferences, if we regard the model as a special case of the fractional endowment exchange model where agents own equal divisions of all objects, then every BTM is equivalent to a simultaneous eating algorithm (SEA) defined by \cite{bogomolnaia2001new}, and every BTM that treats all agents equally at each step is equivalent to the probabilistic serial (PS) mechanism. 
This implies that our extension of BTM subsumes an extension of PS to the full preference domain. Interestingly, our extension of PS can be described as a SEA that differs from the original PS in two respects. First, agents are required to label their consumption as available for others to consume when they find indifferent available objects. Second, when an agent's labeled consumption is being eaten by others, her eating rate is instantly increased according to a rule called ``you request my house—I get your rate.''\footnote{\cite{YuzhangFTTCnew} use this rate-adjusting rule to obtain an extension of PS to the house allocation with existing tenants model under strict preferences.} This rule ensures that an agent whose available consumption is consumed by others is compensated with an equal amount of indifferent objects.

An extreme case of weak preference is dichotomous preference, where agents classify objects as either acceptable or unacceptable and view all acceptable objects as indifferent. \cite{bogomolnaia2004random} propose the egalitarian solution for this environment, which selects an efficient allocation that maximizes equality in agents' welfare. We prove that our extension of PS finds the egalitarian solution under dichotomous preferences.

In the literature, \citet{KS2006} use tools from network flow theory to develop an extension of PS to the full preference domain. In the special case of dichotomous preference, their algorithm is identical to one that computes a lexicographically optimal flow in a parametric network, which produces the egalitarian solution of \citeauthor{bogomolnaia2004random}. Their extension of PS is an iterative application of this algorithm on the full preference domain. \cite{AS2011} develop this tool to obtain an extension of their mechanism for the fractional endowment exchange model to the full preference domain. In contrast, our method does not rely on similar tools from operations research. 

The fractional endowment exchange model reduces to the housing market model when all agents hold integer endowments. Therefore, under strict preferences, every BTM reduces to TTC. This implies that our extension of BTM also subsumes an extension of TTC to the full preference domain. In this case, the definition of BTM resembles the extension of TTC proposed by \cite{jaramillo2012difference}. To preserve strategy-proofness of TTC, \cite{jaramillo2012difference} carefully impose pointing rules in the algorithm. However, since we do not pursue such incentive properties, our definition is intentionally kept general by allowing agents to point to multiple indifferent objects at each step. 

The rest of the paper is organized as follows. Section \ref{section:FEEmodel} presents the fractional endowment exchange model. Section \ref{section:FTTC:fulldomain} presents our extension of BTM. Section \ref{section:property} proves that our extension maintains the efficiency and fairness properties of BTM under strict preferences. Section \ref{section:house:allocation} discusses the extension of PS implied by our method and proves that it finds the egalitarian solution of \cite{bogomolnaia2004random} under dichotomous preferences. Section \ref{section:conclusion} concludes the paper.

\section{The fractional endowment exchange model}\label{section:FEEmodel}

A \textbf{fractional endowment exchange} (FEE) economy is represented by a quadruple $ \G=(I,O,\succsim_I, \w) $ where $ I$ is a finite set of agents, $ O $ is a finite set of objects, $ \succsim_I=(\succsim_i)_{i\in I} $ is a preference profile of agents, and $ \w=(\w_{i,o})_{i\in I,o\in O} \in [0,1]^{I\times O} $ is an endowment matrix. For each $ i\in I $, $ \succsim_i $ represents $ i $'s preferences over $ O $, which is complete and transitive but may not be strict. Let $ \succ_i $ and $ \sim_i $ denote the asymmetric and symmetric components of $ \succsim_i $, respectively. In the matrix $ \w $, $ \w_{i,o} $ represents the amount of object $ o $ owned by agent $ i $; $ \w_i=(\w_{i,o})_{o\in O}$ represents $ i $'s endowments; $ \w_o=(\w_{i,o})_{i\in I}$ represents the distribution of $ o $ among agents' endowments. Each agent demands the equivalent of one unit of an object, and their endowment does not exceed this demand; that is, for all $ i\in I $, $ \sum_{o\in O} \w_{i,o}\le 1 $. Let $ q_o = \sum_{i\in I}\w_{i,o} $ denote the quota of object $ o $ in the economy. In the allocation of indivisible objects, $ \w_{i,o} $ represents the probability share of object $ o $ initially owned by agent $ i $, and $ q_o $ is typically an integer. However, we allow objects to have non-integer quantities. This accommodates another interpretation of the model in which objects are infinitely divisible and might be available in any quantities. This interpretation has been employed by \cite{Bogomolnaia2015random} and others.

\paragraph{Allocation} In an economy $ (I,O,\succsim_I, \w) $, an \textbf{assignment} for agent $ i $ is a vector $ p_i\in [0,1]^{O} $ such that $ \sum_{o\in O}p_{i,o}\le 1 $. The \textbf{size} of $ p_i $ is the total amount of objects it includes, defined as $ \norm{p_{i}}=\sum_{o\in O} p_{i,o}$. The \textbf{support} of $ p_i $ is the set of objects it includes, defined as $ \text{supp}(p_i)=\{o\in O: p_{i,o}>0\}$. An assignment $ p'_i=(p'_{i,o})_{o\in O} $ is a \textbf{subassignment} of $ p_i $ if, for all $ o\in O $, $ p'_{i,o}\le p_{i,o} $. Let $ 2^{p_i} $ denote the set of all subassignments of $ p_i $. 

Given preferences $ \succsim_i $, $p_i$ \textbf{weakly (first-order) stochastically dominates} $p_i'$, denoted by $ p_i\succsim^{sd}_i p'_i $, if $ \sum_{o^{\prime }\succsim _{i}o}p_{i,o^{\prime }}\geq\sum_{o^{\prime }\succsim _{i}o}p'_{i,o^{\prime }}$ for all $ o\in O $. If the inequality is strict for some $ o $, $p_i$ \textbf{strictly stochastically dominates} $p'_i$, denoted by $ p_i \succ^{sd}_i p'_i $. For any $ x\in [0, \norm{p_i}] $,  $ p_i[x] $ denotes a subassignment of $ p_i $ such that $ \norm{p_i[x]}=x $ and $ p_i[x] \succsim^{sd}_i p'_i $ for all $ p'_i\in 2^{p_i} $ with $ \norm{p'_i}=x $. In words, $ p_i[x] $ is the best subassignment of $ p_i $ with size $ x $ for $ \succsim_i $.

An allocation specifies an assignment for each agent, subject to the supply constraints of objects. Since we are considering exchange economies without monetary transfers, we are particularly concerned with allocations resulting from one-for-one exchanges. Therefore, an \textbf{allocation} is represented by a matrix $p=(p_{i,o})_{i\in I, o\in O}\in [0,1]^{I\times O}$ such that, for all $ o\in O$, $\sum_{i\in I}p_{i,o}= q_o$, and, for all $ i\in I$, $\sum_{o\in O}p_{i,o}=\sum_{o\in O}\w_{i,o}$. For each $ i\in I $, $p_{i}=(p_{i,o})_{o\in O}$ is the assignment received by $ i $.  An allocation $p$ weakly stochastically dominates another $p^{\prime }$, denoted by $p \succsim^{sd}_I p^{\prime }  $, if $p_{i}\succsim^{sd}_i p_{i}^{\prime }$ for all $i\in I$; $ p $ strictly stochastically dominates $ p' $, denoted by $p \succ^{sd}_I p^{\prime }  $, if there further exists an agent $ j $ such that $p_{j} \succ^{sd}_j p_{j}^{\prime }$. %If all elements of $ p $ are integers, then $p$ is a \textbf{deterministic} allocation. The Birkhoff-von Neumann theorem and its generalization \citep{birkhoff1946three,von1953certain,kojima2010incentives} ensure that every allocation is a convex combination of deterministic allocations. 

An allocation $ p $ is \textbf{sd-efficient} if it is not strictly stochastically dominated by any other allocation.  It is \textbf{individually rational} (IR) if, for all $ i\in I $, $ p_i \succsim^{sd}_i \w_i $.% It is \textbf{balanced} if, for all $ i\in I $, $ \norm{p_i}=\norm{\w_i} $; that is, agents receive assignments equal in size to their endowments. If an allocation is IR, it must be balanced. 

We define four fairness axioms, ranked from the weakest to the strongest: \textbf{equal treatment of equals} (ETE), \textbf{equal-endowment no envy} (EENE), \textbf{bounded envy} (BE), and \textbf{no envy} (NE). Formally, in an allocation $ p $, agent $ i $ is said to \textbf{envy} another $ j $ if $ p_i\not\succsim^{sd}_i p_j $. Then, $ p $ satisfies
\begin{itemize}
	\item ETE if, for all $ i,j\in I $ such that $ \w_i=\w_j $ and $ \succsim_i=\succsim_j $, $ p_i= p_j $;
	
	\item EENE if, for all $ i,j\in I $ such that $ \w_i=\w_j $,  $ p_i\succsim^{sd}_ip_j $ and $ p_j\succsim^{sd}_jp_i $;
	
	\item BE if, for all $ i,j \in I$, $ \max_{o\in O} \big[\sum_{o'\succsim_i o} p_{j,o'}-\sum_{o'\succsim_i o}p_{i,o'}\big]\le \sum_{o\in O:\w_{j,o}> \w_{i,o}} \big(\w_{j,o}-\w_{i,o} \big) $;
	
	\item NE if, for all $ i,j \in I$, $ p_i\succsim^{sd}_ip_j $ and $ p_j\succsim^{sd}_jp_i $.
\end{itemize}

ETE and EENE are standard fairness axioms in the literature. BE, introduced by \cite{YuzhangFTTCnew}, imposes fairness requirements on any two agents, whether they hold equal or unequal endowments. It requires that if an agent $ i $ envies another $ j $, $ j $'s advantage in the allocation relative to $ i $ must be bounded by her relative advantage in endowments. For instance, if $ j $ owns slightly more endowments than $ i $ but receives a much better assignment, the allocation is considered unfair.  If $ i $ envies $ j $, there must exist object $ o $ such that $ \sum_{o'\succsim_i o} p_{j,o'}>\sum_{o'\succsim_i o}p_{i,o'} $. In the definition of BE, we use $ \max_{o\in O} \big[\sum_{o'\succsim_i o} p_{j,o'}-\sum_{o'\succsim_i o}p_{i,o'}\big] $ to measure the intensity of $ i $'s envy towards $ j $. Note the right-hand side of the inequality sums only over objects $ o $ such that $ \w_{j,o}> \w_{i,o} $. When $ i $ envies $ j $, $ j $ does not need to own more of every object than $ i $'s. For example, suppose that $ i $ owns a unit of $ a $, $ j $ owns $ (1/2b, 1/2c) $, and they all least prefer $ a $. If the two agents obtain their own endowments in an allocation, then $ i $ envies $ j $, and this allocation satisfies BE.

NE is arguably the strongest among all possible fairness axioms. While it is compatible with IR in special cases, they are generally incompatible in the FEE model.

\paragraph{Mechanism} A \textbf{mechanism} $ \psi $ finds an allocation for each economy. When there is no confusion to represent an economy by its preference profile $ \succsim_I$, let $ \psi(\succsim_I) $ denote the allocation found by $ \psi $  and $ \psi_i(\succsim_I) $ the assignment to agent $ i $. A mechanism satisfies an efficiency or fairness axiom if its allocations for all economies satisfy the axiom. 

An agent $ i $ is able to \textbf{manipulate} a mechanism $\psi $ in an economy $ \succ_I $ by reporting $ \succ'_i\neq\succ_i $ if $\psi_{i}(\succsim_{I}) \not\succsim^{sd}_i \psi _{i}(\succsim^\prime_{i},\succsim_{-i})$. An agent $ i $ is able to \textbf{strongly manipulate} $\psi $ in $ \succ_I $ by reporting $ \succ'_i \neq \succ_i $ if $\psi _{i}(\succsim^\prime_{i},\succsim_{-i}) \succ^{sd}_i \psi_{i}(\succsim_{I})$. A mechanism $ \psi $ is \textbf{(weakly) strategy-proof} if it cannot be (strongly) manipulated by any agent in any economy.

Bounded invariance is a property of mechanisms requiring that an agent cannot change the allocation of an object by changing the reported ordering of less preferred objects. Formally, given the set of objects $ O $, for any preferences $ \succ_i $ and any object $ o $, define $ U(\succsim_i,o)=\{o'\in O: o'\succsim_i o\} $ as the upper contour set of $ o $ in $ \succsim_i $. Define $ \succsim_i(o)=\succsim_i|_{U(\succsim_i,o)} $. Then, a mechanism $\psi$ satisfies \textbf{bounded invariance} if, for any economy $ \Gamma= (I,O,\succsim_I, \w) $, any $ i\in I $, any $ o\in O $, and any $ \succsim'_i \neq \succsim_i $ such that $ \succsim_i(o)= \succsim'_i(o)$, we have, for all $ j\in I $, $\psi_{j,o}(\succsim_I)=\psi_{j,o}(\succsim_{I\backslash \{i\}},\succsim'_i)$. Bounded invariance is independent of strategy-proofness.

\paragraph{Classical models} The \textbf{housing market model} is a special case of the FEE model where $ |I|=|O| $, $ q_o=1 $ for all $ o\in O $, and $ \w $ is a permutation matrix (i.e., each agent owns a distinct object). In the \textbf{house allocation model}, all objects are collectively owned by all agents. In this paper, we regard it as a special case of the FEE model where agents own equal divisions of all objects; that is, for all $ i\in I $ and all $ o\in O $, $ \w_{i,o}=\frac{q_o}{|I|}$. For this sake, we assume that objects are no more than agents, that is, $ \sum_{o\in O}q_o \le |I| $. Then, these two classical models are unified within the FEE model.

 %In the simplest form of the \textbf{house allocation model}, $ |I|=|O| $, $ q_o=1 $ for all $ o\in O $, and all objects are collectively owned by all agents. It is regarded as a special case of the FEE model in which agents own equal divisions of all objects; that is, for all $ i\in I $ and all $ o\in O $, $ \w_{i,o}=1/|I|$. 

\section{Balanced trading mechanisms for the FEE model}\label{section:FTTC:fulldomain}

\subsection{Definition of BTM under strict preferences}\label{section:FTTC:strictpref}

Assuming strict preferences, \cite{YuzhangFTTCnew} introduce the class of BTM for the FEE model. At each step, after each remaining agent reports her unique favorite remaining object, a linear equation system characterizes the trading process. The solution to the equation system specifies the amount of her favorite object each agent receives and the amount of each endowment she loses at the step. We select the maximum solution subject to the supply constraints of objects, which is equivalent to letting agents trade as many objects as possible at each step.  Every BTM is IR and sd-efficient, and a subclass of them achieve fairness when the parameters in the equation systems satisfy required conditions. 

To facilitate our discussion, we present the definition under strict preferences. For any economy $ (I,O,\succsim_I, \w) $, at the beginning of each step $ d$, $ I(d) $ denotes the set of remaining agents; $ O(d) $ denotes the set of remaining objects; $ \w(d)=(\w_{i,o}(d))_{i\in I,o\in O} $ denotes the remaining endowments of agents; $ p(d)=(p_{i,o}(d))_{i\in I,o\in O} $ denotes the allocation  found before that step. For each $ i\in I(d) $, $ o_i(d) $ denotes $ i $'s favorite object among $ O(d) $. In the trading process at step $ d $, for each $ i\in I(d) $, $ x_i(d) $ denotes the amount of $ o_i(d) $ that $ i $ receives; for each $ o\in O(d) $, $ x_o(d) $ denotes the amount of $ o $ assigned to agents. Since the assigned amount of $ o $ comes from the endowments of its owners, for each $ i\in I(d) $, we introduce a parameter $ \lambda_{i,o}(d) $ to denote the proportion of $x_o(d) $ that comes from $ i $'s endowment. That is, $ i $ supplies $ \lambda_{i,o}(d) x_o(d) $ of $ o $ to the trading process at step $ d $. Therefore, for all $ o\in O(d) $ and all $ i\in I(d) $,
	$ \lambda_{i,o}(d)\in [0,1] $, $ \sum_{i\in I(d)}\lambda_{i,o}(d)=1 $, and $ \lambda_{i,o}(d)>0 $ only if $ \w_{i,o}(d)>0 $. 
We impose the constraints $ \lambda_{i,o}(d) x_o(d)\le \w_{i,o}(d) $ to ensure that agents do not trade more objects than their endowments.

%Now, we are ready to define the  for the FEE model.

\begin{center}
	\textbf{BTM under Strict Preferences}
\end{center}

\begin{itemize}
	\item[] \textbf{Initialization}: For any $ (I,O,\succsim_I, \w) $,
	$ I(1)=I $, $ O(1)=O $, $ \w(1)=\w$, and $ p(1)=\mathbf{0} $.

	\item[] \textbf{Step} $ d\ge 1 $: Each $ i\in I(d) $ demands her favorite remaining object $ o_i(d) $. Given parameters $ \Lambda(d)=(\lambda_{i,o}(d))_{i\in I(d),o\in O(d)} $ chosen by the mechanism designer, let $ \mathbf{x}^*(d)=(x^*_a(d))_{a\in I(d)\cup O(d)} $ denote the maximum solution to the equation system 
	\begin{equation}\label{equation7}
		\begin{cases}
			x_o(d)=\sum_{i\in I(d)} \mathbf{1}\{o_i(d)=o\} x_i(d)  & \text{for all } o\in O(d),\\
			x_i(d)=\sum_{o\in O(d)} \lambda_{i,o}(d) x_o(d) &  \text{for all }i\in I(d),
		\end{cases}
	\end{equation}
	subject to the constraints
	\begin{equation}\label{equation9}
		\lambda_{i,o}(d) x_o(d) \le \w_{i,o}(d) \quad \text{ for all }i\in I(d) \text{ and }o\in O(d).
	\end{equation}
	
	Let agents trade endowments according to $ \bx^*(d) $. Therefore, for all $ i\in I $ and all $ o\in O $, 
	
	\[
	\w_{i,o}(d+1)=\begin{cases}
		\w_{i,o}(d)-\lambda_{i,o}(d)x_o(d) & \text{ if }i\in I(d) \text{ and }o\in O(d),\\
		0 & \text{ otherwise,}
	\end{cases}
	\]
	and
	\[
	p_{i,o}(d+1)=\begin{cases}
		p_{i,o}(d)+x^*_i(d) & \text{ if }i\in I(d) \text{ and }o=o_i(d),\\
		p_{i,o}(d) & \text{ otherwise.}
	\end{cases}
	\]
	
	Let $ I(d+1)=\{i\in I(d):\sum_{o\in O} \w_{i,o}(d+1)>0\} $ and $ O(d+1)=\{o\in O(d):\sum_{i\in I(d+1)} \w_{i,o}(d+1)>0\} $. If $ O(d+1) $ is empty, stop the algorithm and return $p(d+1)$. Otherwise, go to step $ d+1 $.
\end{itemize}

BTM is a class of mechanisms because varying the equation parameters yields different procedures. Among this class, an particularly appealing mechanism is the so-called \textbf{Equal-BTM} (denoted by $ \psi^{E} $) as it embodies an intuitive fairness principle. It uses the following parameters to ensure that, at each step $ d $, the remaining owners of each remaining object use equal amounts of the object for trading: for all $ i\in I(d) $ and all $ o\in O(d) $,
\begin{equation}\label{equation:equal-BTM}
	\lambda_{i,o}(d)=\begin{cases}
		\dfrac{1}{|j\in I(d):\w_{j,o}(d)>0|} & \text{if }\w_{i,o}(d)>0; \\
		0 & \text{if }\w_{i,o}(d)=0.
	\end{cases}
\end{equation}

\subsection{Extension to the full preference domain}
On the full preference domain, agents may view some objects as indifferent. To accommodate indifferent preferences in the procedure of BTM, our approach is to let agents label their assigned objects as available for further trading if they find indifferent objects among those still available. We then utilize the trading process to ensure that every agent's welfare is unaffected if they provide received objects to the others. To implement this approach, we divide each step of BTM into three stages. Each step begins with the \textbf{labeling} stage, where agents label their received objects for trading. Then, the \textbf{pointing} stage determines agents' demanded objects. Finally, in the \textbf{trading} stage, a linear equation system characterizes the trading process.

At any step, in the labeling stage, the first group of agents who label their received objects are those who find indifferent objects among those still available. After that, a sequence of other agents may be inductively induced to label their received objects as available. Every such sequence can be represented by a chain:
\begin{equation*}\label{chain1}
\color{blue}i_{m+1}\rightarrow \color{black} o_{m}\rightarrow i_m \rightarrow \cdots  o_{k}\rightarrow i_{k}\rightarrow o_{k-1} \rightarrow 
\cdots o_1 \rightarrow i_{1}\rightarrow o_{0},
\end{equation*}
where $ o_0 $ is an object still available in agents' remaining endowments, every other object in the chain has been exhausted in the trading process, and for each $ k=1,\dots,m $, $ i_k $ is an agent who is indifferent between $ o_{k-1} $ and $ o_k $ and has received a positive amount of $ o_{k} $. Since $ i_1 $ is indifferent between $ o_0 $ and $ o_1 $, $ i_1 $ first labels her received amount of $ o_1 $ as available. After that, it induces $ i_2 $ to label her received amount of $ o_2 $ as available, and so on, until $ i_m $ labels her received amount of $ o_m $ as available. Then, in the pointing stage, if an agent $ i_{m+1} $ strictly prefers an object in the chain (say $o_m$) over any other available object, we let $ i_{m+1} $ demand that object. In the trading stage, if $ i_{m+1} $ receives a positive amount of that object, the other agents in the chain are compensated with an equal amount of indifferent objects.% Identifying these sequences of available consumption for trading is crucial for preserving the sd-efficiency of BTM. 

To formally define the mechanisms, in addition to the notations in Section \ref{section:FTTC:strictpref}, we introduce a few others. 
At each step $ d$, $L(d)$ denotes the set of agents who label their received objects as available; for each $i\in L(d)$, $\tilde{O}_i(d)$ denotes the set of objects labeled by $ i $ as available; $\tilde{O}(d)=\cup_{i\in L(d)} \tilde{O}_i(d) $ denotes the set of labeled objects; $ \overline{O}(d)=O(d)\cup \tilde{O}(d) $ denotes the set of available objects, including those available in agents' remaining endowments and those labeled by agents; for each $ i\in I(d) $, $ A_i(d) $ denotes the set of objects demanded by $ i $. We introduce parameters $ \gamma_{i,o}(d)$ to represent how $ i $ divides her demand among the objects in $ A_i(d) $. So, $ \gamma_{i,o}(d)\in [0,1]$ and $ \sum_{o\in A_i(d)} \gamma_{i,o}(d)=1$.\footnote{For instance, if we want $ i $ to demand only one object when there are several indifferent objects, we can set $ \gamma_{i,o}(d)=1 $ for some $ o\in A_i(d) $ and $ \gamma_{i,o'}(d)=0 $ for all other $ o'\in A_i(d) $.}

\begin{center}
	\textbf{BTM on the full preference domain}
\end{center}

\noindent \textbf{Initialization}: For any $ (I,O,\succsim_I, \w) $, $ O(1)=O $, $ \w(1)=\w$, and $ p(1)=\mathbf{0} $.

\noindent \textbf{Step} $ d\ge 1 $: Each step consists of three stages.

\begin{enumerate}
	
	\item \textbf{Labeling}: Agents label their received objects as available  in multiple rounds. 
	
	\begin{itemize}
		\item Round 1: If any $ i\in I $ is indifferent between any $ o \in O(d) $ and any $ o'\in O\backslash O(d) $ such that $ p_{i,o'}(d)>0 $, $ i $ labels her received amount of $ o' $ as available. Let $ \tilde{O}_i (d) $ denote the set of objects labeled by $ i $. Let $ L_1(d) $ denote the set of such agents $ i $. Let $ \tilde{O}_1(d)=\cup_{i\in L_1(d)}\tilde{O}_i (d)  $. 
		
		\item Round 2: If any $ i\in I\backslash L_1(d) $ is indifferent between any $ o\in \tilde{O}_1(d) $ and any $ o'\in O\backslash [O(d)\cup \tilde{O}_1(d)] $ such that $ p_{i,o'}(d)>0 $, $ i $ labels her received amount of $ o' $ as available.\footnote{It is without loss of generality to exclude $L_1(d)$ from Round 2 because if any $i\in L_1(d)$ wants to label multiple objects in her consumption as available, she must have done so in Round 1. An agent will never label two objects in her consumption that are not indifferent to her as available at the same step.} Let $ \tilde{O}_i (d) $ denote the set of objects labeled by $ i $. Let $ L_2(d) $ denote the set of such agents $ i $. Let $ \tilde{O}_2(d)=\cup_{i\in L_2(d)}\tilde{O}_i (d)  $.
		
		\item[]$ \vdots $ 
		
		\item Round n: If any $ i\in I\backslash \cup_{k=1}^{n-1}L_k(d) $ is indifferent between any $ o\in \tilde{O}_{n-1}(d) $ and any $ o'\in O\backslash [O(d)\cup \tilde{O}_1(d) \cup \cdots \cup \tilde{O}_{n-1}(d) ] $ such that $ p_{i,o'}(d)>0 $, $ i $ labels her received amount of $ o' $ as available. Let $ \tilde{O}_i (d) $ denote the set of objects labeled by $ i $. Let $ L_n(d) $ denote the set of such agents $ i $. Let $ \tilde{O}_n(d)=\cup_{i\in L_n(d)}\tilde{O}_i (d)  $.
	\end{itemize}
 
	Since there are finite agents and finite objects and each agent appears in at most one round, the Labeling stage must stop in finite rounds. Suppose that it stops in $ n $ rounds. Let $ L(d)=\cup_{k=1}^{n}L_k(d) $, $ \tilde{O}(d)=\cup_{i\in L(d)} \tilde{O}_i(d) $, and $ \overline{O}(d)=O(d)\cup \tilde{O}(d) $.
	
	\item \textbf{Pointing}: Define $ I(d)=L(d)\cup \{i\in I:\sum_{o\in O}\w_{i,o}(d)>0 \} $. 
	\begin{itemize}
		\item Round 1: For each $ i\in I(d) $, if $ i $'s favorite objects among $ \overline{O}(d) $ include objects from $ O(d) $, let $ A_i(d) $ denote the set of her favorite objects from $ O(d) $. Let $ P_1(d) $ denote the set of such agents $ i $.
		
		\item Round 2: For each $ i\in I(d)\backslash  P_1(d) $, if $ i $'s favorite objects among $ \overline{O}(d)  $ include objects from $ \tilde{O}_1(d) $, let $ A_i(d) $ denote the set of her favorite objects from $ \tilde{O}_1(d) $. Let $ P_2(d) $ denote the set of such agents $ i $.
		
		\item[]$ \vdots $
		
		\item Round m: For each $ i\in I(d)\backslash [ \cup_{k=1}^{m-1} P_k(d)] $, if $ i $'s favorite objects among $ \overline{O}(d)  $ include objects from $ \tilde{O}_{m-1}(d) $, let $ A_i(d) $ denote the set of her favorite objects from $ \tilde{O}_{m-1}(d) $. Let $ P_m(d) $ denote the set of such agents $ i $.			 
	\end{itemize}
	
	Since there are finite agents and finite objects, the Pointing stage must stop in finite rounds. Note that, for each $ i\in I(d) $, $A_i(d)\cap \tilde{O}_i(d)=\emptyset$.

	\item \textbf{Trading}: Given parameters $ \lambda(d)= \big(\lambda_{i,o}(d)\big)_{i\in I(d),o\in \overline{O}(d)} $ and $ \gamma(d)= \big(\gamma_{i,o}(d)\big)_{i\in I(d),o\in A_i(d)}$,\footnote{For each $ o\in \tilde{O}(d) $ and $ i\in I(d) $, $ \lambda_{i,o}(d)\in [0,1] $, $ \sum_{i\in I(d)}\lambda_{i,o}(d)=1 $, and $ \lambda_{i,o}(d)>0 $ only if $o\in \tilde{O}_i(d)$.} let $ \mathbf{x}^*(d)=(x^*_a(d))_{a\in I(d-1)\cup \overline{O}(d-1)} $ denote the maximum solution to the equation system:	
   \begin{equation}\label{equation:system}
   	\begin{cases}
   	x_o(d)=\sum_{i\in I(d):o\in A_i(d)} \gamma_{i,o}(d)x_i(d) & \text{ for all } o\in  \overline{O}(d),\\
   	x_i(d)=\sum_{o\in \overline{O}(d)} \lambda_{i,o}(d) x_o(d) & \text{ for all }i\in I(d),\\
   	\end{cases}
   \end{equation}
   subject to the constraints
   \begin{equation}\label{constraints}
   	\begin{cases}
   	\lambda_{i,o}(d) x_o(d) \le \w_{i,o}(d) & \text{for all }  i\in I(d)  \text{ and all } o\in O(d), \\
   	\lambda_{i,o}(d) x_o(d) \le p_{i,o}(d) & \text{for all }  i\in I(d)  \text{ and all } o\in \tilde{O}(d).
   	\end{cases}
   \end{equation}

Let agents trade endowments according to $ \bx^*(d) $. Therefore, for all $ i\in I$ and all $ o\in O $,
	\[
	\w_{i,o}(d+1)=\begin{cases}
	\w_{i,o}(d)-\lambda_{i,o}(d)x_o(d) & \text{ if }i\in I(d) \text{ and }o\in O(d),\\
	0 & \text{ otherwise,}
	\end{cases}
	\]
	and
	\[
	p_{i,o}(d+1)=\begin{cases}
	p_{i,o}(d)-\lambda_{i,o}(d)x_o(d) & \text{ if }i\in I(d) \text{ and }o\in \tilde{O}_i(d),\\
	p_{i,o}(d)+\gamma_{i,o}(d)x_i(d) & \text{ if }i\in I(d) \text{ and }o\in A_i(d),\\
	p_{i,o}(d) & \text{ otherwise.}
	\end{cases}
	\]
	%For all $ i\in I\backslash I(d) $, $ \w_i(d+1)=\w_i(d) $  and $ p_i(d+1)=p_i(d) $.
	
	Let $ O(d+1)=\{o\in O(d):\sum_{i\in I} \w_{i,o}(d+1)>0\} $. If $ O(d+1) $ is empty, stop the algorithm. Otherwise, go to step $ d+1 $.
\end{enumerate}

Lemma 1 of \cite{YuzhangFTTCnew} ensures that the maximum solution $ \mathbf{x}^*(d) $ at each step exists and is non-negative. By varying the parameters $\lambda(d)$ and $\gamma(d)$, we obtain a class of BTM on the full preference domain. For instance, any BTM using the parameters in (\ref{equation:equal-BTM}) is viewed as an extension of $ \psi^E $. These extensions differ in the selection of $ (\gamma_{i,o}(d))_{i\in I(d),o\in A_i(d)} $ and $ (\lambda_{i,o}(d))_{i\in L(d),o\in \tilde{O}(d)} $. For example, we may choose  $ \gamma_{i,o}(d)=\frac{1}{|A_i(d)|} $ for each $ i\in I(d) $ and each $ o\in A_i(d) $, and $ \lambda_{i,o}(d) =\frac{1}{|\{j\in L(d):o\in \tilde{O}_j(d)\}|}$ for each $ i\in L(d) $ and each $ o\in \tilde{O}_i(d) $.

We illustrate our approach with an example. The example explains how $ \psi^E $ is extended to accommodate weak preferences using the above parameters.

\begin{example}
Consider four agents $ \{1,2,3,4\} $ and three objects $ \{a,b,c\} $. Agents own equal endowments $ (1/4a,1/4b,1/4c) $  and have the following preferences. Agent 1 is indifferent between $ a $ and $ b $, and agent 2 is indifferent between $ a $ and $ c $. We explain how $ \psi^E $ is extended to accommodate these preferences.
\begin{table}[!h]
	\centering
	\begin{tabular}{cccc}
		$ \succsim_1 $ & $ \succsim_2 $ & $ \succsim_3 $ & $ \succsim_4 $ \\\hline
		$ \{a,b\} $ & $ \{a,c\} $ & $ a $ & $ b $\\
		$ c $ & $ b $ & $ c $ & $ c $\\
		 &  & $ b $ & $ a $
	\end{tabular}
	%\caption{}\label{table:illustration}
\end{table}

\begin{itemize}
	\item Step one: All agents demand their favorite objects. Suppose that 1 and 2 equally divide their demands over their two favorite objects. Then, by solving the equations of $ \psi^E $ that require all agents to lose equal endowments in the trading process, 1 receives $ 1/4a $ and $ 1/4b $, 2 receives $ 1/4a $ and $ 1/4c $, 3 receives $ 1/2a $, and 4 receives $ 1/2b $. Each agent loses $ (1/4a,3/16b,1/16c) $ from their endowment. After this step, $ a $ is exhausted from agents' endowments.
	
	\item Step two: Since 1 is indifferent between $ a $ and $ b $, and $ b $ is still available,  1 labels her received $ 1/4a $ as available for trading. Similarly, 2 also labels her received $ 1/4a $ as available for trading. 	
	Then, 1 demands $ b $, 2 demands $ c $, 3 demands $ a $, and 4 demands $ b $. Suppose that 1 and 2 use equal amounts of $ a $ for trading. By solving the equations of $ \psi^E $, 1 receives $ 3/20b $ and loses $ 1/20a $, 2 receives $ 3/20c $ and loses $ 1/20a $, 3 receives $ 1/10a $, and 4 receives $ 1/10b $. Each agent loses $ (1/16b,3/80c) $ from their endowment. After this step, $ b $ is exhausted from agents' endowments.

	\item Step three: Since 2 is indifferent between $ a $ and $ c $, and $ c $ is still available, 2 labels her remaining $ 1/5a $ as available for trading. Then, since 1 is indifferent between $ a $ and $ b $, this induces 1 to label her received $ 2/5b $ as available for trading. Then, 1 and 3 demand $ a $, 2 demands $ c $, and 4 demands $ b $. Similarly as above, by solving the equations of $ \psi^E $, 1 receives $ 2/15a $ and loses $ 1/15b $, 2 receives $ 4/15c $ and loses $ 3/15a $, 3 receives $ 1/15a $, and 4 receives $ 1/15b $. Each agent loses $ 1/15c $ from their endowment. After this step, the $ 1/5a $ labeled by 2 is exhausted.
	
	\item Step four: No agents label their received objects as available for trading. Therefore, all agents demand the only remaining object $ c $. By solving the equations of $ \psi^E $, each agent receives $ 1/12c $, and loses $ 1/12c $ from their endowment. After this step, all objects are exhausted.
\end{itemize}

In the final allocation, 1 receives $(1/3a,1/3b,1/12c)$, 2 receives $ 3/4c $, 3 receives $(2/3a,1/12c)$, and 4 receives  $(2/3b,1/12c)$. This allocation satisfies IR, sd-efficiency, and NE. 

This example can be viewed as a house allocation problem since agents own equal divisions of objects. Therefore, this example  demonstrates how we extend PS to accommodate weak preferences. The trading process at each step can be interpreted as an eating process where agents' eating rates follow the rule introduced in Section \ref{section:house:allocation}. 
\end{example}

\section{Properties of BTM on the full preference domain}\label{section:property}

We prove that our extension of BTM to the full preference domain maintains their efficiency and fairness properties under strict preferences.

We first prove that $ \overline{O}(d)$, the set of available objects at step $ d $,  weakly shrinks as the algorithm proceeds. It means that once an object becomes unavailable at any step, it remains unavailable at subsequent steps. This fact is crucial for the properties of BTM.  

\begin{lemma}\label{lemma1}
	For any step $ d $ of BTM on the full preference domain, $ \overline{O}(d+1)\subseteq \overline{O}(d) $.
\end{lemma}

To facilitate our discussion, at each step $d$ of BTM, we define 	$ x^c_i(d)=\sum_{o\in \tilde{O}(d)} \lambda_{i,o}(d) x_o(d) $, and $ x^n_i(d)=\sum_{o\in O(d)} \lambda_{i,o}(d) x_o(d) $. 
That is, $ x^c_i(d) $ is the amount of objects $ i $ loses from her labeled objects at step $ d $; $ x^n_i(d) $ is the amount of objects $i$ loses from her endowments. We have $ x_i(d)= x^c_i(d)+x^n_i(d)$, and $ x^n_i(d) $ equals $ i $'s net consumption at step $ d $.

\begin{proposition}\label{section5:prop1}
	Every BTM on the full preference domain is IR and sd-efficient.
\end{proposition}

\cite{YuzhangFTTCnew} show that by controlling the equation parameters in the procedure of BTM, the resulting allocation can achieve fairness axioms defined in Section \ref{section:FEEmodel}.  We show that these results still hold on the full preference domain.

\begin{definition}
	A BTM  satisfies
	\begin{itemize}
		\item[(1)]  \textbf{stepwise equal treatment of equals} (stepwise ETE) if, at each step $ d $, for each $i,j\in I(d)$, 
	$[ \w_i(d)=\w_j(d), \tilde{O}_i(d)=\tilde{O}_j(d), \text{and }A_i(d)=A_j(d)]\implies$ $\lambda_{i,o}(d)=\lambda_{j,o}(d) $ for all $ o\in \overline{O}(d) $.
		
		\item[(2)]  \textbf{stepwise equal-endowment equal treatment} (stepwise EEET) if, at each step $ d $, for each $i,j\in I(d)$, 	
		$ \w_i(d)=\w_j(d)  \implies $ $\lambda_{i,o}(d)=\lambda_{j,o}(d) $ for all $ o\in O(d) $.
		
		\item[(3)] \textbf{bounded advantage} if, at each step $ d $, for each $i,j\in I(d)$ and each $o\in O(d)$, 	
		$ \w_{i,o}(d)\ge \w_{j,o}(d)>0 \implies \frac{\w_{i,o}(d)}{\w_{j,o}(d)} \ge \frac{\lambda_{i,o}(d)}{\lambda_{j,o}(d)} \ge 1 $.
	\end{itemize}
\end{definition}

 At any step $ d $, for any $i,j\in I(d)$, if $\lambda_{i,o}(d)=\lambda_{j,o}(d) $ for all $ o\in \overline{O}(d) $, then $ x^c_i(d)=x^c_j(d) $ and $ x^n_i(d)=x^n_j(d) $. Stepwise ETE requires that, at the beginning of any step $ d $, as long as two agents have equal remaining endowments, label the same set of available objects, and demand the same set of favorite objects , they are treated equally at the step. Stepwise EEET imposes a weaker requirement on a larger set of agents. For any two agents who have equal remaining endowments at the beginning of any step $ d $, by requiring $\lambda_{i,o}(d)=\lambda_{j,o}(d) $ for all $ o\in O(d) $, it ensures that $ x^n_i(d)=x^n_j(d) $ and $ \w_i(d+1)=\w_j(d+1) $.
 
 Bounded advantage generalizes stepwise EEET. At each step $ d $, if $ \w_{i,o}(d)\ge \w_{j,o}(d) >0$, then $\frac{\lambda_{i,o}(d)}{\lambda_{j,o}(d)} \ge 1 $ implies $ \lambda_{i,o}(d)x_o(d)-\lambda_{j,o}(d)x_o(d) \ge 0 $, which means that $ i $ uses weakly more of $ o $ than $ j $ does for trading at step $ d $; while $ \frac{\w_{i,o}(d)}{\w_{j,o}(d)} \ge \frac{\lambda_{i,o}(d)}{\lambda_{j,o}(d)} $ implies $ \lambda_{i,o}(d)x_o(d)-\lambda_{j,o}(d)x_o(d) \le  \w_{i,o}(d)-\w_{j,o}(d)$,\footnote{$ \lambda_{i,o}(d)x_o(d)-\lambda_{j,o}(d)x_o(d)= \lambda_{j,o}(d)x_o(d)\big(\frac{\lambda_{i,o}(d)}{\lambda_{j,o}(d)}-1\big) \le \w_{j,o}(d) \big(\frac{\w_{i,o}(d)}{\w_{j,o}(d)}-1\big)= \w_{i,o}(d)-\w_{j,o}(d)$.} which means that $ i $'s relative advantage in using $ o $ for trading is bounded by her relative advantage in the endowment of $ o $.

\begin{proposition}\label{prop:generalfairness}
	(1) A BTM satisfying stepwise ETE satisfies ETE;
	
	(2) A BTM satisfying stepwise EEET satisfies EENE;
	
	(3) A BTM satisfying bounded advantage satisfies BE.
\end{proposition}

Any extension of $ \psi^E $ is an example of BTM that satisfies bounded advantage.

Our definition of BTM is very flexible, allowing mechanism designers to choose equation parameters contingent on various factors. When preferences are strict, \cite{YuzhangFTTCnew} define a subclass of BTM, called \textbf{regular} BTM, where the parameters at each step depend only on the distribution of agents' endowments. Every regular BTM is anonymous, neutral, and satisfies EENE and bounded invariance. This definition extends to the full preference domain. For the new parameter $ \gamma_{i,o}(d) $, we can similarly require that its selection depends only on the distribution of available objects at the step. Then, the desirable properties of regular BTM are preserved. We omit the details since they are straightforward.

\cite{YuzhangFTTCnew} prove more precise fairness properties of $ \psi^E $ in simple economies where each agent owns at most one object. We prove that these properties are preserved for the extensions of $ \psi^E $. Formally, an economy $ (I,O,\succsim_I, \w) $ is called \textbf{simple} if, for each $ i\in I $, there exists at most one $ o\in O $ such that $ \w_{i,o}>0 $. In a simple economy, an allocation $ p $ satisfies \textbf{ordinal fairness} if, for each object $ o $ and any two owners $ i $ and $ j $ with $ 0<\w_{i,o}\le \w_{j,o} $:

(1) $ p_{i,a}>0\implies\sum_{o'\succsim_i a}p_{i,o'}\le  \sum_{o'\succsim_j a}p_{j,o'} $, and

(2) $ p_{j,a}>0\implies\sum_{o'\succsim_j a}p_{j,o'}\le  \sum_{o'\succsim_i a}p_{i,o'} $, unless $ \sum_{o'\succsim_i a}p_{i,o'}= \norm{p_i}<\sum_{o'\succsim_j a}p_{j,o'}$.

Ordinal fairness implies \textbf{generalized EENE}: for each object $ o $ and its any two owners $ i $ and $ j $ with $ 0<\w_{i,o}\le \w_{j,o} $, $ p_j \succsim^{sd}_j p_i $, and for all  $ p'_j \in 2^{p_j} $ with $ \norm{p'_j}=\norm{p_i} $, $ p_i \succsim^{sd}_i p'_j $.

\begin{proposition}\label{prop:ordinalfairness}
	In simple FEE economies, every extension of $ \psi^E $ satisfies ordinal fairness.
\end{proposition}

\section{House allocation model}\label{section:house:allocation}

In the house allocation model with strict preferences, every BTM is equivalent to a SEA introduced by \cite{bogomolnaia2001new}, and every BTM satisfying stepwise EEET is equivalent to PS. Proposition \ref{prop:generalfairness} implies that every BTM satisfying stepwise EEET defined in this paper is an extension of PS that maintains its efficiency and fairness properties. 

\begin{corollary}
	In the house allocation model, every BTM satisfying stepwise EEET is an extension of PS to the full preference domain that satisfies sd-efficiency and NE.
\end{corollary}

In the house allocation model, every BTM satisfying stepwise EEET can be described as a SEA: at any time $ t\in [0,1] $, agents label their consumption of exhausted objects as available for others to consume if they find indifferent available objects in the economy; then, agents consume their favorite available objects with the following rates:
\[
s_i(t)=1+\sum_{o\in \tilde{O}_i(t)}\big[\lambda_{i,o}(t)\sum_{j\in I:o\in A_j(t)}\gamma_{j,o}(t)s_j(t)\big].
\]
In words, in addition to the unit rate, every agent's eating rate is instantly increased by an amount equal to the total rate at which her labeled consumption is being consumed by the others. We call this rate-adjusting rule ``\textit{you request my house - I get your rate}''.

\subsection{Dichotomous preferences}\label{section:house_allocation_dich}

The dichotomous preference domain is an extreme case of weak preferences. Under dichotomous preferences, each agent's welfare in an allocation is measured by the total amount of acceptable objects she receives. \cite{bogomolnaia2004random} define the \textit{egalitarian solution} to achieve efficiency and maximal fairness simultaneously.\footnote{\cite{bogomolnaia2004random} also prove that any mechanism selecting the egalitarian solution is group strategy-proof, meaning that no group of agents can jointly manipulate the mechanism.} \citet{KS2006} use tools from network flow theory to generate assignments in the egalitarian solution. In this subsection, we prove that in the special case of dichotomous preferences, every BTM satisfying stepwise EEET is an algorithm that finds the egalitarian solution.

Formally, we consider a house allocation allocation model in which each $ i\in I $ classifies the objects in $ O $ into the set of acceptable objects (denoted by $ C_i $) and the set of unacceptable objects, and is indifferent between the objects in the same class. Each object is acceptable to at least one agent. The Gallai-Edmonds Decomposition Lemma of \cite{bogomolnaia2004random} states that every such problem can be decomposed into three subproblems. In the first, a perfect match exists where each agent obtains an acceptable object. In the second, there is an oversupply of acceptable objects for any subset of agents, so that every agent can also obtain an acceptable object. In the third, there is a shortage of acceptable objects for any subset of agents. The egalitarian solution is proposed for the third subproblem. We restrict attention to the third subproblem by assuming that, for every nonempty $ O'\subseteq O $, $|\{i\in I:C_i\cap O'\neq \emptyset\}| >|O'| $. 

For every nonempty $ Y\subseteq I $ and nonempty $ O'\subseteq O $, define $\Gamma(Y, O')=\big(\cup_{i\in Y}C_i\big)\cap O'$, which is the set of objects among $ O' $ that are acceptable to at least one agent in $ Y $.

The egalitarian solution is characterized by a sequence of bottleneck sets of agents with their assigned objects.  The first bottleneck set $ X^*_1 $ satisfies
\begin{equation}\label{bottleneck}
X^*_1=\arg\min_{Y\subseteq I}\frac{|\Gamma(Y,O)|}{|Y|}.
\end{equation}
When there are multiple solutions to the above problem, let $ X^*_1 $ be the solution of the largest cardinality.\footnote{Since the union of two solutions to the problem is still a solution, $ X^*_1 $ is unique.} In the egalitarian solution, $ \Gamma(X^*_1, O) $ are assigned to $ X^*_1 $ such that each $ i\in X^*_1 $ obtains $ \frac{|\Gamma(X^*_1,O)|}{|X^*_1|} $ of acceptable objects. 

The second bottleneck set is found similarly among the remaining agents $ Z_1= I\backslash X^*_1 $ and the remaining objects $ P_1=O\backslash \Gamma(X^*_1,O) $. 
Inductively, for each $ k\ge 2 $, the $ k $-th bottleneck set $ X^*_k $ satisfies
\begin{equation*}
X^*_k=\arg\min_{Y\subseteq Z_{k-1}}\frac{|\Gamma(Y,P_{k-1})|}{|Y|},
\end{equation*}
where $ Z_{k-1}=I\backslash (\cup_{\ell=1}^{k-1} X^*_\ell) $ and $P_{k-1}=P_{k-2}\backslash \Gamma(X^*_{k-1},P_{k-2})$.
When there are multiple solutions, let $ X^*_k $ be the solution of the largest cardinality. In the egalitarian solution, $ \Gamma(X^*_k,P_{k-1}) $ are assigned to $ X^*_k $ such that each $ i\in X^*_k $ obtains $ \frac{|\Gamma(X^*_k,P_{k-1})|}{|X^*_k|} $ of acceptable objects.

Let $ (X^*_1, X^*_2,\ldots,X^*_m) $ denote the sequence of bottleneck sets in the egalitarian solution.

We prove that the sequence of bottleneck sets defined above is generated in the procedure of every BTM satisfying stepwise EEET. Specifically, by Lemma \ref{lemma1}, $ \overline{O}(d+1)\subseteq \overline{O}(d) $ for each step $ d $. Now, we call a step $d$ \textbf{critical} if $\overline{O}(d)\backslash \overline{O}(d+1)\neq \emptyset $. That is, some objects available at the beginning of step $d$ become unavailable at the end of step $d$. Let $ (d_1,d_2,\ldots,d_n) $ denote the sequence of critical steps in the procedure of our mechanism. For each critical step $ d_k $, define $ X_{k}= \{i\in I: p_{i,o}(d_{k}+1)>0 \text{ for some } o\in \overline{O}(d_{k})\backslash \overline{O}(d_{k}+1) \}$. That is, $ X_k $ is the set of agents who receive positive amounts of some objects from $ \overline{O}(d_{k})\backslash \overline{O}(d_{k}+1) $. 
Lemma \ref{lemma:bottleneck} proves that $ X_k=X^*_k $. That is, $ X_k $ is a bottleneck set of agents defined above. 

\begin{lemma}\label{lemma:bottleneck}
	In every house allocation problem with dichotomous preferences, the number of critical steps in the procedure of any BTM satisfying EEET equals the number of bottleneck sets of agents defined by \cite{bogomolnaia2004random}, and for every critical step $ d_k $,
	\begin{center}
		$ X_k=X^*_k $, and $ \overline{O}(d_k)\backslash \overline{O}(d_k+1)=\Gamma(X^*_k,P_{k-1}) $
	\end{center}
 where $ P_{0}=O $ and $ P_{k}=P_{k-1}\backslash \Gamma(X^*_{k},P_{k-1}) $.
\end{lemma}

Therefore, we obtain the following result. 

\begin{proposition}\label{prop:dich}
	In every house allocation problem with dichotomous preferences, every BTM satisfying stepwise EEET finds the egalitarian solution of \cite{bogomolnaia2004random}.
\end{proposition}

\section{Conclusion}\label{section:conclusion}

Our approach to dealing with weak preferences can be applied to a broad range of problems involving trading mechanisms to find desirable allocations. The key idea is to utilize a trading process that enables agents to exchange consumption while preserving efficiency. Depending on the application, this process may vary in complexity. However, the method developed by \cite{YuzhangFTTCnew} allows us to define the trading process using linear equations, avoiding potential complications. We hope to see further applications of our method in future research.

\bibliographystyle{aea}

\setlength{\bibsep}{0pt plus 0.3ex}
\bibliography{reference}

\clearpage

\appendix

\section*{Appendix: Proofs}

\begin{proof}[\normalfont \textbf{Proof of Lemma \ref{lemma1}}]
	By definition, $ \overline{O}(d)=O(d)\cup \tilde{O}(d) $. Since $ O(d+1)\subseteq O(d)\subseteq \overline{O}(d) $, to prove that $ \overline{O}(d+1)\subseteq \overline{O}(d) $, it suffices to prove $ \tilde{O}(d+1)\subseteq \overline{O}(d) $. We enumerate the objects in $ \tilde{O}(d+1) $. Define $ E(d)= O(d)\backslash O(d+1) $, which is the set of objects exhausted from agents' endowments at step $ d $. Suppose there are $ n $ rounds in the labeling stage of step $ d+1 $. Then, $ \tilde{O}(d+1)=\cup_{k=1}^{n}\tilde{O}_k(d+1) $. We then prove by induction that $ \tilde{O}(d+1)\subseteq \overline{O}(d) $.
	
	\textit{Base step.} For each $ o'\in \tilde{O}_1(d+1) $, if $ o'\in E(d) $, then $ o'\in O(d) \subseteq \overline{O}(d) $. If $ o'\notin E(d) $, then $ o'\notin O(d) $. So $ o' $ is exhausted before step $ d $. Then, $ o'\in \tilde{O}_1(d+1) $ means that some agent $ i $ is indifferent between $ o' $ and some $ o\in O(d+1) $. If $ i $ receives a positive amount of $ o $ at the end of step $ d-1 $, then, since $ o\in O(d+1)\subseteq O(d) $, $ i $ must label her received amount of $ o' $ as available at step $ d $. If $ i $ does not receive a positive amount of $ o' $ at the end of step $ d-1 $, then $ i $ must receive a positive amount of $ o' $ during step $ d $. Since $ o' $ is exhausted before step $ d $, some agent must label her received amount of $ o' $ as available at step $ d $. Therefore, in any case, $ o'\in \overline{O}(d) $. Thus,  $ \tilde{O}_1(d+1)\subseteq \overline{O}(d) $. 
	
	\textit{Induction step.} Suppose that, for all  $ \ell=1,\ldots,k-1 $, $ \tilde{O}_\ell(d+1)\subseteq \overline{O}(d) $. We want to prove that $ \tilde{O}_k(d+1)\subseteq \overline{O}(d) $. For each $ o'\in \tilde{O}_k(d+1) $, if $ o'\in E(d) $, then $ o'\in O(d) \subseteq \overline{O}(d) $. If $ o'\notin E(d) $, then $ o'\notin O(d) $. So $ o' $ is exhausted before step $ d $. Then, $ o'\in \tilde{O}_k(d+1) $ means that some agent $ i $ is indifferent between $ o' $ and some $ o\in \tilde{O}_{k-1}(d+1) $. If $ i $ receives a positive amount of $ o' $ at the end of step $ d-1 $, then, since $ o\in \tilde{O}_{k-1}(d+1)\subseteq O(d) $, $ i $ must label her received amount of $ o' $ as available at step $ d $. If $ i $ does not receive a positive amount of $ o' $ at the end of step $ d-1 $, then $ i $ must receive a positive amount of $ o' $ during step $ d $. Since $ o' $ is exhausted before step $ d $, some agent must label her received amount of $ o' $ as available at step $ d $. So, in any case, $ o'\in \overline{O}(d) $. Thus, $ \tilde{O}_k(d+1)\subseteq \overline{O}(d) $. 
	
	By induction, for all $ 1\le \ell \le n $, $ \tilde{O}_\ell(d+1)\subseteq \overline{O}(d) $. So, $ \tilde{O}(d+1)\subseteq \overline{O}(d) $.
\end{proof}

\begin{proof}[\normalfont \textbf{Proof of Proposition \ref{section5:prop1}}]
	(IR) At each step $ d $ of BTM, each agent $ i $ receives a net amount $ x_i^n(d) $ of her favorite object among available objects and loses an equal amount of endowments. Therefore, at each step, the net assignment each agent receives weakly stochastically dominates the endowments she loses. In aggregation, the final assignment received by each agent weakly stochastically dominates her endowments.

	(Sd-efficiency) Suppose that for some preference profile, the allocation found by some BTM is not sd-efficient. Then, there must exist $ k\geq 2 $ agents, who are denoted by $ i_1, i_2, \ldots, i_k $ and do not need to be distinct, and $ k $ objects in the assignments they receive, denoted by $ o_1 , o_2,\ldots, o_k $, such that if the $ k $ agents exchange a positive amount of the $ k $ objects as represented by the following cycle, none of them is worse off and some is strictly better off:
	$$i_{1}\rightarrow  o_2\rightarrow i_{2}\rightarrow  o_3\rightarrow i_{3}\rightarrow \cdots \rightarrow  o_k\rightarrow i_{k}\rightarrow  o_1\rightarrow i_{1}.$$
	By trading objects following the cycle, every $ i_\ell $ obtains an amount of $ o_{\ell+1} $ from $i_{\ell+1}$'s consumption and loses an equal amount of $o_\ell$ from her own consumption to $i_{\ell-1}$. Without loss of generality, assume that $ i_1 $ is strictly better off after the trade. This means that $ i_1 $ strictly prefers $ o_2 $ to $ o_1 $. Suppose that in the procedure of BTM, $ i_1 $ starts demanding $o_1 $ at step $ d $. Then, it must be $ o_1\in \overline{O}(d) $ and $ o_2\notin \overline{O}(d) $. We then consider $ i_2 $. Assume that $ i_2 $ starts demanding $ o_2 $ at step $ d' $. There are two cases:
	
	\begin{itemize}
		\item  If $ i_2 $ strictly prefers $ o_3 $ to $ o_2 $, then it must be $ o_2\in \overline{O}(d') $ and $ o_3\notin \overline{O}(d') $. Since $ o_2\notin \overline{O}(d) $, by Lemma 1, it must be that $ d'<d $. So, $ o_3\notin \overline{O}(d) $.
		
		\item If $ i_2 $ is indifferent between $ o_3 $ and $ o_2 $, since $ o_2\notin \overline{O}(d) $, it must be $ o_3\notin \overline{O}(d) $, since otherwise $ i $ should label her received amount of $ o_2 $ as available.
	\end{itemize}
	
	In any case, we have $ o_3\notin \overline{O}(d) $. We can inductively apply the above arguments to the remaining agents and objects in the cycle, which leads to the conclusion that $ o_1\notin \overline{O}(d) $. This is a contradiction.
\end{proof}

\begin{proof}[\normalfont \textbf{Proof of Proposition \ref{prop:generalfairness}}]
	
	(1) For any two agents $ i $ and $ j $ such that $ \w_i=\w_j $ and $ \succsim_i=\succsim_j $, at step one of BTM, they must demand the same set of objects. By stepwise ETE, $\lambda_{i,o}(1)=\lambda_{j,o}(1) $ for all $ o\in \overline{O}(1) $, implying that they obtain equal amounts of each object and lose equal amounts of each endowment at step one. At the beginning of step two, they must still have equal remaining endowments, label the same set of available objects, and demand the same set of favorite objects. By stepwise ETE, $\lambda_{i,o}(2)=\lambda_{j,o}(2) $ for all $ o\in \overline{O}(2) $, implying that they obtain equal amounts of each object and lose equal amounts of each endowment at step two. This inductively holds for all remaining steps. So, in the final allocation, they must receive equal assignments.
	
	(2) For any two agents $ i $ and $ j $ with $\w_i=\w_j$, at step one, by stepwise EEET, $\lambda_{i,o}(1)=\lambda_{j,o}(1) $ for all $ o\in O(1) $. It ensures that $ x^n_i(1)=x^n_j(1) $ and $ \w_i(2)=\w_j(2) $. At step two, by stepwise EEET, we still have  $\lambda_{i,o}(2)=\lambda_{j,o}(2) $ for all $ o\in O(2) $, ensuring that $ x^n_i(2)=x^n_j(2) $ and $ \w_i(3)=\w_j(3) $. This inductively holds for all remaining steps. So, at each step, the two agents obtain equal net amounts of their respective favorite objects. It implies that they do not envy each other in the final allocation.
	
	(3) In any economy, let $ p $ be the outcome of any BTM satisfying bounded advantage. 
	Suppose that an agent $ i $ envies another $ j $. Let $ o^* $ be the solution to $ \max_{o\in O} \big[\sum_{o'\succsim_i o} p_{j,o'}-\sum_{o'\succsim_i o}p_{i,o'}\big] $.
	Let $ d $ be the earliest step after which all objects in $ U(\succsim_i, o^*) $ become unavailable. That is, $ U(\succsim_i, o^*)\cap \overline{O}(d+1)=\emptyset $ and $ U(\succsim_i, o^*)\cap \overline{O}(d)\neq\emptyset $. By Lemma \ref{lemma1}, $U(\succsim_i, o^*)\cap \overline{O}(d')=\emptyset $ for all $ d'\ge d+1 $. Then, $ \sum_{o\succsim_i o^*}p_{i,o}=\sum_{d'=1}^d x^n_i(d')=\sum_{o\in O}\big(\w_{i,o}-\w_{i,o}(d+1)\big) $ and $ \sum_{o\succsim_i o^*}p_{j,o}\le\sum_{d'=1}^d x^n_j(d')=\sum_{o\in O}\big(\w_{j,o}-\w_{j,o}(d+1)\big) $. So, 
	\[
	\sum_{o\succsim_i o^*}p_{j,o}-\sum_{o\succsim_i o^*}p_{i,o}
	\le\sum_{o\in O}\big[\big(\w_{j,o}-\w_{j,o}(d+1)\big)-\big(\w_{i,o}-\w_{i,o}(d+1)\big)\big].
	\]
	
	For every $ o\in O $ such that $ \w_{i,o}\ge \w_{j,o} $, bounded advantage implies that, for all $ 1\le d'\le d $, $ \w_{i,o}(d'+1)\ge \w_{j,o}(d'+1) $ and $ \lambda_{i,o}(d')\ge \lambda_{j,o}(d') $. So, 
	$
	\w_{j,o}-\w_{j,o}(d+1)=\sum_{d'=1}^d \lambda_{j,o}(d')x^*_o(d')\le  \sum_{d'=1}^d \lambda_{i,o}(d')x^*_o(d')=\w_{i,o}-\w_{i,o}(d+1).
	$
	Equivalently,
	\[
	\big(\w_{j,o}-\w_{j,o}(d+1)\big)-\big(\w_{i,o}-\w_{i,o}(d+1)\big)\le  0.
	\] 
	
	For every $ o\in O $ such that $ \w_{i,o}< \w_{j,o} $, bounded advantage implies that, for all $ 1\le d'\le d $, $ \w_{i,o}(d'+1)\le \w_{j,o}(d'+1) $ and $ \lambda_{i,o}(d')\le \lambda_{j,o}(d') $. In particular, $ \w_{i,o}(d+1)\le \w_{j,o}(d+1)  $. So,
	\[
	\big(\w_{j,o}-\w_{j,o}(d+1)\big)-\big(\w_{i,o}-\w_{i,o}(d+1)\big)\le \w_{j,o}-\w_{i,o}.
	\]
	
	Therefore, $
	\sum_{o\succsim_i o^*}p_{j,o}-\sum_{o\succsim_i o^*}p_{i,o}
	\le \sum_{o\in O:\w_{i,o}< \w_{j,o}} \big(\w_{j,o}-\w_{i,o} \big)$.
\end{proof}

\begin{proof}[\normalfont \textbf{Proof of Proposition \ref{prop:ordinalfairness}}]
	In any simple FEE economy, let $ p $ denote the allocation found by any extension of $ \psi^E $. Consider any object $ o $ and its any two owners $ i $ and $ j $ with $0< \w_{i,o}\le \w_{j,o} $. Then, $ \norm{p_i}=\w_{i,o} $ and $ \norm{p_j}=\w_{j,o} $. 
	
	For any object $ a $ such that $ p_{i,a}>0$, let $ d $ be the earliest step after which all objects in $ U(\succsim_i,a) $ become unavailable; that is, $ U(\succsim_i,a)\cap \overline{O}(d)\neq \emptyset $ and $ U(\succsim_i,a)\cap \overline{O}(d+1)=\emptyset $. Then, $ \sum_{o'\succsim_i a}p_{i,o'}=\sum_{d'=1}^d x^n_i(d') $. Object $ a $ must be available at step $ d $; that is, $ a\in \overline{O}(d) $. If $ a\notin O(d) $, then $ U(\succsim_i,a)\cap \overline{O}(d)\neq \emptyset $ implies the existence of available objects that $ i $ regards as indifferent with $ a $, which leads $ i $ to label her consumption of $ a $ as available. 	
	At each step $ d' $ with $ 1\le d' \le d $, the two agents obtain their respective favorite objects, and $ x^n_j(d')= x^n_i(d')$. Since $ a\in \overline{O}(d) $, by Lemma \ref{lemma1}, $ a\in \overline{O}(d') $. It means that, at each step $ d' $, $ j $ obtains weakly better objects than $ a $. Therefore, $ \sum_{o'\succsim_j a}p_{j,o'}\ge \sum_{d'=1}^d x^n_j(d')=\sum_{d'=1}^d x^n_i(d')= \sum_{o'\succsim_i a}p_{i,o'}$.

	Similarly, for any object $ a $ such that $ p_{j,a}>0$, if $ \sum_{o'\succsim_j a}p_{j,o'}\le \norm{p_i} $, we can switch the roles of $ i $ and $ j $ in the above argument to prove $\sum_{o'\succsim_j a}p_{j,o'}\le  \sum_{o'\succsim_i a}p_{i,o'} $. If $ \sum_{o'\succsim_j a}p_{j,o'}> \norm{p_i} $, let $ d $ be the earliest step after which all objects in $ U(\succsim_i,a) $ become unavailable. Then, similarly as above, $ a\in \overline{O}(d)  $. By Lemma \ref{lemma1}, $ a\in \overline{O}(d')  $ for all $ d'<d $.  Since $ \sum_{o'\succsim_j a}p_{j,o'}> \norm{p_i} $, $ i $ must exhaust her endowment before step $ d $. Therefore, at each step where $ i $ obtains a positive amount of an object, the object must be weakly better than $ a $ for $ i $. This means that $ \sum_{o'\succsim_i a}p_{i,o'} =\norm{p_i} $.
\end{proof}

\begin{proof}[\normalfont \textbf{Proof of Lemma \ref{lemma:bottleneck}}]
	We prove the lemma by induction.
	
	\textit{Base case.} We first prove that $ \overline{O}(d_1)\backslash \overline{O}(d_1+1)=\Gamma(X_1,O)$. Because all objects in $ \overline{O}(d_1)\backslash \overline{O}(d_1+1) $ are assigned to agents at the end of step $ d_1 $, $ \overline{O}(d_1)\backslash \overline{O}(d_1+1)\subseteq \Gamma(X_1,O) $. Suppose that there exists $ i\in X_1 $ and $ o\in C_i $ such that $ o\notin \overline{O}(d_1)\backslash \overline{O}(d_1+1) $. Since $ d_1 $ is the first critical step, $ \overline{O}(d_1)=O $. Therefore, $ o\in \overline{O}(d_1+1) $. But this means that $ i $ should label her consumption of the objects from $ \overline{O}(d_1)\backslash \overline{O}(d_1+1) $ as available at step $ d_1+1 $, which is a contradiction. Therefore, $
	\overline{O}(d_1)\backslash \overline{O}(d_1+1)= \Gamma(X_1,O)$.
	
	We then prove that $X_1=X^*_1$. Define $t_1=\frac{|\Gamma(X_1,O)|}{|X_1|}$.  Stepwise EEET implies that all agents obtain equal amounts of net consumption at each step. Since no objects become unavailable before step $ d_1 $, at the end of step $ d_1 $, the amount of acceptable objects each agent obtains must be equal to $ t_1 $. 
	For any nonempty $ Y\subseteq X_1 $, $ \Gamma(Y,O)\subseteq \Gamma(X_1,O)= \overline{O}(d_1)\backslash \overline{O}(d_1+1)$. Therefore, $\frac{|\Gamma(Y,O)|}{|Y|}=t_1$. 
	For any nonempty $ Y\subseteq I $ such that $ Y\backslash X_1\neq \emptyset $, it must be that, for each $ j\in Y\backslash X_1 $, $ C_j\cap \overline{O}(d_1+1)\neq \emptyset $, because otherwise $ j\in X_1 $. Since each $ i\in Y $ obtains $ t_1 $ of acceptable objects at the end of step $ d_1 $ and each $ j\in Y\backslash X_1 $ has acceptable objects that are still available after step $ d_1 $, it must be that
	$\frac{|\Gamma(Y,O)|}{|Y|}>t_1$. Therefore, 
	\[
	X_1=\arg\min_{Y\subseteq I}\frac{|\Gamma(Y,O)|}{|Y|},
	\]and $ X_1 $ is the solution of the largest cardinality. This means that $ X_1=X^*_1 $.
	
	\textit{Induction step.} Suppose that for all $ \ell=1,\ldots,k-1 $, $X_\ell=X^*_\ell$, and $ \overline{O}(d_\ell)\backslash \overline{O}(d_\ell+1)=\Gamma(X^*_\ell,P_{\ell-1})$. We want to prove that $X_k=X^*_k$ and $ \overline{O}(d_{k})\backslash \overline{O}(d_{k}+1)=\Gamma(X^*_{k},P_{k-1}) $, where $ P_{k-1}= \overline{O}(d_{k})$. Define $ Z_{k-1}=I\backslash (\cup_{\ell=1}^{k-1} X^*_\ell) $. Since the agents among $ \cup_{\ell=1}^{k-1} X^*_\ell $ receive objects only from $ O\backslash \overline{O}(d_{k}) $, $ X_{k}\subseteq Z_{k-1} $ and $ \overline{O}(d_{k})\backslash \overline{O}(d_{k}+1)\subseteq \Gamma(X_{k},\overline{O}(d_{k})) $.
	
	Suppose that there exists $ i\in X_{k} $ and $ o\in C_i $ such that $ o\in \overline{O}(d_{k}+1) $. Then, $ i $ should label her consumption of the objects from $ \overline{O}(d_{k})\backslash \overline{O}(d_{k}+1) $ as available at the beginning of step $ d_{k}+1 $, which is a contradiction. Therefore, $\overline{O}(d_{k})\backslash \overline{O}(d_{k}+1)= \Gamma(X_{k},\overline{O}(d_{k}))$.  
	Then, we can use arguments similar to those in the base case to prove that
	\[
	X_{k}=\arg\min_{Y\subseteq Z_{k-1}}\frac{|\Gamma(Y,\overline{O}(d_{k}))|}{|Y|},
	\] 
	and $ X_{k} $ is the solution of the largest cardinality. It means that $ X_{k} =X^*_{k} $.
\end{proof}

\begin{proof}[\normalfont \textbf{Proof of Proposition \ref{prop:dich}}]
	It is implied by Lemma \ref{lemma:bottleneck}.
\end{proof}

\end{document}